%% file: main.tex
\documentclass[letterpaper, 10 pt, conference]{ieeeconf}  

\IEEEoverridecommandlockouts                              
\input{notation}

\newcommand{\inter}{{\rm inter}} %
\newcommand{\intra}{{\rm intra}} %
\newcommand{\ctr}{{\rm ctr}} %
\newcommand{\dir}{{\rm dir}} %
\newcommand{\blkdiag}{{\rm blkdiag}} %

\newcommand\SmallMatrix[1]{{%
		\scriptsize\arraycolsep=0.4\arraycolsep\ensuremath{\begin{bmatrix}#1\end{bmatrix}}}}

\newcommand\SmallMatrixNBla[1]{{%
		\small\arraycolsep=0.9\arraycolsep\ensuremath{\begin{matrix}#1\end{matrix}}}}

\newcommand\scalemath[2]{\scalebox{#1}{\mbox{\ensuremath{\displaystyle #2}}}}

\newcommand{\EP}{\SmallMatrixNBla{\frac{1}{\varepsilon}}} %

\newcommand{\EPt}{\SmallMatrixNBla{\frac{t}{\varepsilon}}} %

\newcommand\PC[1]{^{(#1)}}

\usepackage{cite}

\newcommand*{\QEDA}{\null \hfill\ensuremath{\triangle}}

\title{\LARGE \bf
Vibrational Control of Cluster Synchronization: \\Connections with Deep Brain Stimulation
}

\author{Yuzhen Qin, Danielle S. Bassett, and Fabio Pasqualetti
	\thanks{This work was supported in part by awards ARO-W911NF1910360 and NSF-NCS-FO-1926829. Y. Qin and F. Pasqualetti are with the Department of Mechanical
		Engineering, University of California at Riverside, \{yuzhenqin, fabiopas\}@engr.ucr.edu. D. S. Bassett is with
		the Department of Bioengineering, the Department of Electrical \&
		Systems Engineering, the Department of Physics \& Astronomy, the
		Department of Psychiatry, and the Department of Neurology, University
		of Pennsylvania, and The Santa Fe Institute, dsb@seas.upenn.edu.
	}
}

\begin{document}

\maketitle
\thispagestyle{empty}
\pagestyle{empty}

\begin{abstract}
	Cluster synchronization underlies various functions in the brain. Abnormal patterns of cluster synchronization are often associated with neurological disorders. Deep brain stimulation (DBS) is a neurosurgical technique used to treat several brain diseases, which has been observed to regulate neuronal synchrony patterns. Despite its widespread use, the mechanisms of DBS remain largely unknown. In this paper, we hypothesize that DBS plays a  role similar to \emph{vibrational control} since they both highly rely on high-frequency excitation to function. Under the framework of Kuramoto-oscillator networks, we study how vibrations introduced to network connections can stabilize cluster synchronization. We derive some sufficient conditions and also provide an effective approach to design vibrational control. Also, a numerical example is presented to demonstrate our theoretical findings. 
\end{abstract}

\section{INTRODUCTION}

Cluster synchronization describes a phenomenon where the behavior of the units in a network evolves into different, yet synchronized, clusters. It has been widely observed in the brain as correlated neural activity. Patterns of cluster synchronization underlie various brain functions such as neuronal communication, memory formation, and cognition \cite{HG-PAA-DG-AA-KA:19,FJ-AN:2011}. Pathological synchrony patterns also characterize many brain disorders, e.g., Parkinson's disease \cite{Hammond2007} and epilepsy \cite{JP-DCM-JJGR-SCA:2013}. 



DBS is a therapeutic technique that uses intracranial electrodes and high-frequency electrical stimulation (see Fig.~\ref{conceptual} (a)) to treat such brain disorders \cite{Krauss2021}. Despite its ubiquitous use, the underlying mechanisms of DBS are still elusive. Selecting the proper dose of stimulation and the best location in the brain, which is crucial for achieving the greatest clinical benefits of DBS, still largely relies on trial and error \cite{Krauss2021}. It is thus critical to further characterize the mechanisms that underpin  DBS's efficacy, which, in turn, can inform the design of optimal and personalized brain stimulation. 

In this paper, we provide a hypothetical mechanism for DBS from the perspective of control systems. We propose that high-frequency signals delivered by DBS to brain regions work as \textit{vibrational control}. Vibrational control is an open-loop control strategy that has been used to stabilize engineering systems, e.g., inverted pendulums, chemical reactors, and under-actuated robots (see \cite{SMM:80,REB-JB-SMM:86b,BS-BTZ:97,CX-TY-MI:2018} and the references therein). Similar to DBS, vibrational control relies on high-frequency dithers. It is particularly useful in situations where feedback control is not applicable since  online measurement of systems' states or outputs is unavailable.
Using the Kuramoto-oscillator framework, we study how vibrational control can stabilize cluster synchronization. Despite its simplicity, the Kuramoto model has been widely used to study synchronization  in neural systems \cite{BM-HS-DA:2010}. 
 
\textbf{Related work.} Cluster synchronization has recently attracted extensive attention. Some studies have found that cluster synchronization is closely related to network symmetries \cite{Pecora2014,YSC-TN-AEM:17,YQ-MC-BDOA-DSB-FP:20} and equitable partitions \cite{Schaub2016}. Stability conditions for cluster synchronization are constructed in networks with dyadic \cite{TM-GB-DSB-FP:18,QY-KY-PO-CM:21,FP-SA-MT:2020,KR-IH:2021} and hyper connections \cite{SA-DRM:2021a,SA-DRM:2021b}. Different control strategies, such as pinning control \cite{WW-WZ-TC:09}, and intervention of network connections or nodes' dynamics \cite{GLV-FM-LV:2018,DF-GR-MDB:17,TM-GB-DSB-FP:21-arxiv}, are proposed to control cluster synchronization. By contrast, we employ vibrational control, a much more realistic control strategy especially for neural oscillation regulation since it does not need to modify structural brain networks or the intrinsic dynamics of neurons. 


 \begin{figure}[t]
	\centering
	\includegraphics[scale=0.8]{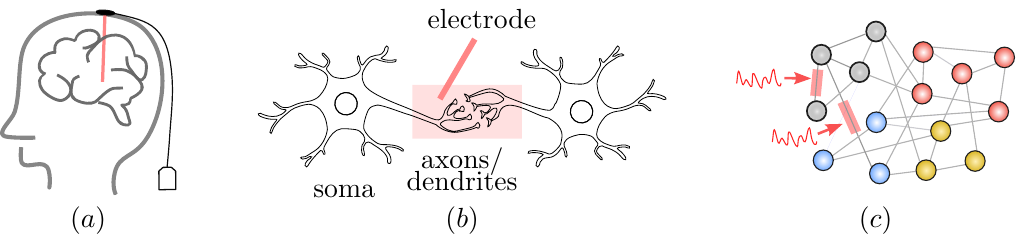}
	\caption{Vibrational control of cluster synchronization inspired by DBS. (a) DBS consists of an intracranial electrode, an extension wire, and a pulse generator. (b) DBS predominantly affects axons and dendrites rather than the soma. (c) Vibrations are introduced to network connections to stabilize cluster synchronization (different colors represent different clusters).}
	\label{conceptual}	
	\vspace{-12pt}
\end{figure}

\textbf{Paper contribution.}  The contribution of this paper is threefold. First, we use vibration control, a potential mechanism of DBS, to regulate cluster synchronization. 
Different from the existing computational studies on DBS (e.g., see \cite{WG-DB-CH-BP-BC-BR:2019}), we introduce control inputs to network connections rather than nodes' dynamics (see Fig.~\ref{conceptual} (c)). This is inspired by the observation that brain stimulation predominantly affects dendrites and axons near the electrode, rather than the soma  \cite{HMH-BV-RRC:09} (see Fig.~\ref{conceptual} (b)). Second, we derive sufficient conditions on the vibrational inputs such that they can stabilize originally unstable cluster synchronization. The key idea is that vibrations are required improve the robustness of synchronization within clusters to a level that is sufficient to overcome perturbations resulting from inter-cluster connections. Averaging methods and perturbation theory are used for the analysis. Third, we propose an analytically tractable approach to design vibrational control. Finally, we provide an example to demonstrate our theoretical findings. 
 
\textbf{Notation.}  We denote the unit circle as $\BS^1$, a point of which is a phase. The $n$-dimensional torus is denoted by $\BT^n = \mathbb{S}\times\dots\times\mathbb{S}$. Given a set $\CS$, $|\CS|$ stands for its cardinality. For two sets $\CS_1$ and $\CS_2$, $\CS_1\backslash\CS_2$ is their set-theoretic difference. Given $A\in\R^{n\times m}$, $B:=[A]^+$ is the matrix where $b_{ij}=a_{ij}$ if $a_{ij}>0$ and $b_{ij}=0$ otherwise. Given matrices $A_1,\dots,A_r$, $\blkdiag(A_1,\dots,A_r)$ is the block-diagonal matrix of them.








\section{Problem Formulation and Preliminary}

\subsection{Problem formulation}


We consider $n$ oscillators coupled by an undirected network $\CG=(\CV,\CE)$ associated with the weighted adjacency matrix $A=A^\top=[a_{ij}]_{n \times n}$.  The dynamics of the oscillators are governed by the following Kuramoto model:
\begin{equation}\label{no_input}
	\dot \theta_i =\omega_i + \sum_{j=1}^{n} a_{ij} \sin(\theta_j- \theta_i),
\end{equation}
where $\theta_i \in \BS^1$ is the $i$th oscillator's phase, and $\omega_i\in \R$ is its natural frequency.  Let $\theta:=[\theta_1,\dots,\theta_n]^\top$ and $\omega:=[\omega_i,\dots,\omega_n]^\top$.

Since we are interested in studying cluster synchronization, we define cluster synchronization manifolds as follows.
\begin{definition}[\textit{Synchronization manifold}]\label{CS:manifold}
	For the the network $\CG=(\CV,\CE)$, consider the partition $\CC:=\{\CC_1,\CC_2,\dots,\CC_r\}$, where $\CC_k\subset \CV$, $\CC_k \cap  \CC_\ell =\emptyset$ for any $k\neq \ell$, and $\cup_{k=1}^r \CC_k=\CV$. The \textit{cluster synchronization manifold} associated with the partition $\CC$ is defined as
	\begin{equation*}
	\CM:= \{\theta \in \BT^n: \theta_i=\theta_j, \forall i,j \in \CC_k,  k= 1,\dots,r\}.\hspace{12pt}\QEDA
\end{equation*}
\end{definition}

The cluster synchronization manifold is invariant along the system \eqref{no_input} if starting from any $\theta(0)\in \CM$, the solution to \eqref{no_input} satisfies $\theta(t) \in \CM$ for all $t\ge 0$. Following \cite{TM-GB-DSB-FP:18,LT-CF-MI-DSB-FP:17}, we make the following assumption to ensure that $\CM$ is invariant. 

\begin{assumption}[\textit{Invariance}]\label{invariance}
 For $k=1,2,\dots, r$:  i) the natural frequencies satisfy $\omega_i=\omega_j$ for any $i,j \in \CC_k$; and ii) the coupling strengths satisfy that, for any $\ell\in\{1,2,\dots,r\}\backslash\{k\}$,  $\sum_{p\in \CC_\ell}(a_{ip}-a_{jp})=0$ for any $i,j \in \CC_k$. \QEDA
\end{assumption}

In addition to Assumption~\ref{invariance}, to guarantee that the cluster synchronization represented by $\CM$ can appear in the network, $\CM$ is required to be stable. Further conditions combining the connection weights and natural frequencies need to be satisfied to ensure the stability of $\CM$ (e.g., some sufficient conditions are constructed in {\cite{TM-GB-DSB-FP:18,QY-KY-PO-CM:21,Schaub2016}). 
	
Yet, changes on the anatomical connectivity or the intrinsic dynamics of neuronal populations, caused by brain aging or disorders \cite{WM-JS-YY:2016},  may violate such conditions so that cluster synchronization patterns needed for normal brain functions are no longer stable. In this paper, we aim to investigate how vibrational control, a control strategy that resembles DBS \cite{Krauss2021}, can restore the stability of such desired synchrony patterns.


Specifically, we consider the following controlled model
\begin{equation}\label{model:controlled}
	\scalemath{1}{\dot \theta_i =\omega_i + \sum_{j=1}^{n} \Big(a_{ij} + q(a_{ij})u_{ij}(t) \Big) \sin(\theta_j- \theta_i)}, 
\end{equation}
where $q(a_{ij})=1$ if $a_{ij}>0$ and $q(a_{ij})=0$, otherwise\footnote{This setting ensures that control inputs can only introduce vibrations to the existing connections and do not unreasonably create new ones.}, and $u_{ij}(t)$ is the corresponding vibrational control input. Although the couplings are symmetric, satisfying  $a_{ij}=a_{ji}$ for any $i$ and $j$, we allow the effect of the vibrational control on each edge to be asymmetric, that is, it is allowed that $u_{ij}(t)\neq u_{ji}(t)$.
As typically done in the literature {(e.g., \cite{BS-BTZ:97,REB-JB-SMM:86b,SMM:80})}, we assume that $u_{ij}(t)$ is periodic with period $T$ and has high frequency and zero mean. It can be seen that $\int_{0}^{T} a_{ij} + q(a_{ij}) u_{ij}(t) dt =a_{ij}$, which means that the effect of vibrations on network weights is zero in average.  Particularly, we consider sinusoidal vibrations in this paper, i.e.,
\begin{align*}
 	u_{ij}(t)=\frac{u_{ij}}{\varepsilon}\sin\Big(\frac{t}{\varepsilon}\Big),
\end{align*}
 where $\varepsilon>0$ determines the amplitude and frequency of the vibrations. Then, the period becomes $T=2\varepsilon\pi$.

It remains to design the parameters of the vibrational control, i.e., $u_{ij}$'s and $\varepsilon$, to stabilize the cluster synchronization manifold $\CM$. To this end, it is necessary to ensure that the vibrational inputs preserve the invariance of $\CM$. In this paper, we consider a special configuration of vibrations that satisfies
\begin{align*}
	&{u_{ij}(t)=0,} &\forall i\in \CC_k, j\in \CC_\ell, \forall k\neq \ell,
\end{align*}
so that the invariance of $\CM$ is preserved.
In other words, vibrations are \textit{only} introduced to the connections between oscillators {within the same clusters} (such connections are referred to as \textit{intra-cluster connections}).


\subsection{Preliminary}
To facilitate analysis in the reminder of this paper, we introduce some notations (summarized in Fig.~\ref{notation}) and derive a compact model for the system \eqref{model:controlled}. 

 \begin{figure}[t]
	\centering
	\includegraphics[scale=1.45]{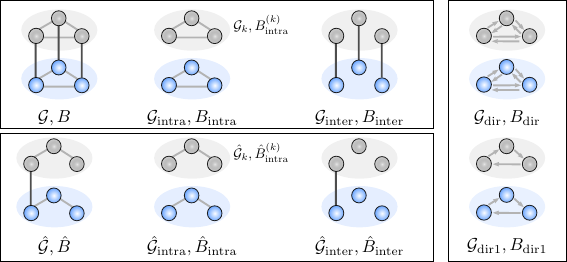}
	\caption{Graphic illustration of the main notations we employ in this paper: different subgraphs and their corresponding incidence matrices. }
	\label{notation}
\end{figure}

For the network $\CG=(\CV,\CE)$ and a partition $\CC:=\{\CC_1,\CC_2,\dots,\CC_r\}$, define $\CG_k=(\CC_k,\CE_k)$ where $\CE_k:=\{(i,j)\in \CE: i,j\in \CC_k\}$. For each $k$, denote $n_k:=|\CC_k|$ as the number of nodes in $\CG_k$. We assume that each $\CG_k$ is connected and contains at least 2 nodes. Let $\CG_\intra=(\CV,\CE_\intra)=\cup_{k=1}^r \CG_k$ and we refer to it as the intra-cluster graph. Let $\CG_{\inter}=(\CV,\CE_\inter)$ be the inter-cluster subgraph, where $\CE_\inter :=\CE\backslash \CE_\intra$. Define the directed graph associated with $\CG_\intra$ by $\CG_\dir=(\CV,\CE_\dir)$, where $\CE_\dir$ is defined in a way such that every edge in $\CG_\intra$ becomes two directed edges in $\CE_\dir$. Let $\CG_{\dir1}=(\CV,\CE_{\dir1})$ be the graph where $\CE_{\dir1}$ contains only one directed edge between each pair of nodes in $\CE_\dir$. Let $\CG_{\dir2}=(\CV,\CE_{\dir2})$ with  $\CE_{\dir2}:=\CE_{\dir}\backslash \CE_{\dir1}$.
Let $\hat \CG=(\CV,\hat \CE)$ be any spanning tree of $\CG$, and thus $|\hat \CE|=n-1$. For each $k$, Let $\hat \CG_{k}=(\CC_k,\hat \CE_k)$ with  $\hat \CE_k:=\hat \CE\cap \CE_k$. The intra-cluster subgraph of  the spanning tree $\hat \CG$ is $\hat \CG_\intra=(\CV,\hat \CE_\intra):=\hat \cup_{k=1}^r \hat \CG_k$;  the inter-cluster subgraph  of $\hat \CG$ is $\hat \CG_\inter=(\CV,\hat \CE_\inter)$ with $\hat \CE_\inter:=\hat \CE\backslash \hat \CE_\intra$. 

We further define the corresponding (oriented) incidence matrices of the above graphs and subgraphs as in Fig.~\ref{notation}. Without loss of generality, we order the columns of these incidence matrices in a way such that  
\begin{equation}\label{matrices}
	\scalemath{0.9}{\begin{matrix*}[l]		
		B=[B_\intra, B_\inter],&B_\intra=\blkdiag( B^{(1)}_{\intra},\dots, B^{(r)}_{\intra})\\
		\hat B=[\hat B_\intra, \hat B_\inter],&\hat B_\intra=\blkdiag(\hat B^{(1)}_{\intra},\dots,\hat B^{(r)}_{\intra}),\\
		B_\dir=[B_{\dir1},B_{\dir2}]&B_{\dir1}=B_\intra,\hspace{12pt}B_{\dir2}=-B_\intra.
	\end{matrix*}}
\end{equation}

Let $W=\scalemath{0.7}{\begin{bmatrix}
		W_\intra&0\\
		0&W_\inter
	\end{bmatrix}
	}$, where $\scalemath{0.75}{W_\intra:=\diag\{a_{ij},(i,j)\in \CG_\intra\}}$ and $W_\inter:=\diag\{a_{ij},(i,j)\in \CG_\inter\}$ are diagonal weight matrices of $ \CG_\intra$ and $ \CG_\inter$, respectively. Similarly, since vibrations are only introduced to the intra-cluster connections and can have asymmetric influence on each of them, let $U(t):=\diag\{u_{ij}\sin(t),(i,j)\in \CE_\dir\}$. It holds that $U(t)=\blkdiag(U_1(t),U_2(t))$ for $U_1(t):=\diag\{u_{ij}\sin(t),(i,j)\in \CE_{\dir1}\}$ and $U_2(t):=\diag\{u_{ij}\sin(t),(i,j)\in \CE_{\dir2}\}$. Further, $U_1$ and $U_2$ can be decomposed as $U_1(t)=\blkdiag(U^{(1)}_1(t),\dots,U^{(r)}_1(t))$ and $U_2(t)=\blkdiag(U^{(1)}_2(t),\dots,U^{(r)}_2(t))$, where $U^{(k)}_1(t))$ and $U^{(k)}_2(t))$ represent the vibrations to the $k$th cluster.

Let $x:=\hat B_\intra^\top \theta$ and $y := \hat B_\inter^\top \theta$, and they define the intra- and inter-cluster phase differences, respectively.  The dynamics of $x$ and $y$ are (see Appendix~\ref{sec:compact_form} for the derivation and the expressions of the matrices $R_1,R_2$, and $R_3$) 
\begin{subequations}\label{compact_form}
\begin{align}
	&\dot x=f_\intra(x) +f_\inter(x,y)+\EP f_\ctr\big( U\big(\EPt\big),x\big),\label{compact_form:1}\\
	&\dot y=g(x,y)+\EP g_\ctr( U(\EPt),x),
\end{align}
\end{subequations}
where 
 \begin{equation}\label{functions}
	\begin{aligned}
		&\scalemath{0.8}{f_\intra(x)=-\hat B_\intra^\top B_\intra W_\intra \sin(R_1 x)},\\
		&\scalemath{0.8}{f_\inter(x,y)= - \hat B_\intra^\top B_\inter W_\inter \sin(R_2 x+R_3y)},\\
		& \scalemath{0.8}{f_\ctr(U,x)=-\hat B_\intra^\top \Big([B_{\intra}]^+U_{1}(t)-[-B_{\intra}]^+U_{2}(t)\Big)\sin(R_1x)}\\
		&\scalemath{0.8}{g(x,y)=\hat B_\inter^\top \omega -\hat B_\inter^\top B_\intra W_\inter \sin(R_1x)}\\
		&\hspace{2cm} \scalemath{0.8}{-\hat B_\inter^\top B_\inter W_\inter \sin(R_2x+R_3y)},\\
		& \scalemath{0.8}{g_\ctr(U,x)=-\hat B_\inter\Big([B_{\intra}]^+U_{1}-[-B_{\intra}]^+U_{2}\Big)\sin(R_1x)},
	\end{aligned}
\end{equation}
Note that $f_\intra$, $f_\inter$, and $f_\ctr$ describe the dynamics induced by the intra connections, inter connections, and vibrational control inputs, respectively. Here, we single out the important properties of these functions that will be used later: i)  $f_\intra(0)=0$, ii) $f_\inter(0,y)=0$ for any $y$, iii) $f_\ctr(0,x)=0$ for any $x$, and (iv)   $f_\ctr(U,0)=0$ for any $U$.

Notice that $x=0$ corresponds to the cluster synchronization manifold $\CM$. The manifold $\CM$ is exponentially stable along the system \eqref{model:controlled} if $x=0$ is exponentially stable uniformly\footnote{The equilibrium $x=0$ is said to be exponentially stable uniformly in $y$ if it is exponentially stable starting from any $y(0)$  \cite[Chap. 4]{haddad2011nonlinear}.} in $y$ along \eqref{compact_form}. To stabilize the cluster synchronization manifold $\CM$, it suffices to design vibration control inputs to ensure the uniform exponential stability of  $x=0$  for the system \eqref{compact_form}. 

\section{Vibrational Control: General Results}

In this section, we derive some general results on vibrational stabilization of $x=0$ of the system \eqref{compact_form}. 

The term $f_\inter(x,y)$ in \eqref{compact_form:1} can be taken as a vanishing perturbation dependent of $y$ to the controlled nominal system
\begin{align}\label{nominal}
	\dot x=f_\intra(x)+\EP f_\ctr( U(\EPt),x).
\end{align}
It can be decomposed as $f_\inter=[(f^{(1)}_{\inter})^\top,\dots,(f^{(r)}_{\inter})^\top]^\top$, where $f^{(k)}_{\inter}=- (\hat B^{(k)}_{\intra})^\top B_\inter W_\inter \sin(R_2 x+R_3y)$ describes the inter-cluster dynamics of the $k$th cluster. 


Next, we show how the vibrational control term $f_\ctr(U,x)$ can stabilize $x=0$ in the presence of the perturbation $f_\inter(x,y)$. To this end, we consider a partially linearized model of the system \eqref{compact_form:1}, which is 
\begin{align}\label{linearized}
	\scalemath{1}{\dot x = (J + \EP P(\EPt))x+f_\inter(x,y)},
\end{align}
where
\begin{equation}\label{expres:JandP}
	\begin{aligned}
		&{J=\frac{\partial f_\intra}{\partial x}(0)= \blkdiag (J\PC{1},\dots,J\PC{r})}, \text{and} \\
		& {P(t)=\frac{\partial f_\ctr}{\partial x}(0)=  \blkdiag (P\PC{1}(t),\dots,P\PC{r}(t))},
	\end{aligned}
\end{equation}
with ${J\PC{k}=-(\hat B^{(k)}_{\intra})^\top B_\intra^{(k)} W^{(k)}_\intra R_1}$ and 
\begin{align}\label{expre:P}
	\scalemath{0.8}{P\PC{k}(t)=-(\hat B^{(k)}_{\intra})^\top \Big([B^{(k)}_{\intra}]^+  U^{(k)}_{1}(t)-[-B^{(k)}_{\intra}]^{+} U^{(k)}_{2}(t)\Big)R_1^{(k)}},
\end{align}
for $k=1,\dots,r$. Note that only the first and third terms of Eq.~\eqref{compact_form:1} are linearized in Eq.~\eqref{linearized}. Observe that $P(t)$ is periodic and has the same period $T$ as $U(t)$.

\begin{lemma}[\textit{Connecting the stability of Systems} \eqref{compact_form} \textit{and} \eqref{linearized}] \label{lemma:linearized}
	If the equilibrium $x=0$ is exponentially stable uniformly in $y$ for the system \eqref{linearized}, it is also exponentially stable uniformly in $y$  for the system  \eqref{compact_form}. 
\end{lemma}

The proof can be found in Appendix~\ref{proof:linearized}. From this lemma, one can see that to stabilize the cluster synchronization manifold $\CM$, it suffices to configure the vibrational control such that $P(t)$ stabilizes  $x=0$ of the system \eqref{linearized}. 

Let $s={t}/{\varepsilon}$. The system \eqref{linearized} can be rewritten as 
\begin{align}\label{change:time:scale}
	\frac{d x}{d s} = (\varepsilon J +  P(s))x+\varepsilon f_\inter( x, y).
 \end{align}

Next, we use averaging methods to analyze this system. However, the standard first-order averaging is not applicable here. Recall that $P(s)$ has zero mean. Then, applying the first-order averaging to \eqref{change:time:scale} just eliminates the $P(s)$ term and results in the uncontrolled system $\frac{d x}{d s} = \varepsilon J x+\varepsilon f_\inter( x, y)$. 

To avoid that, we change the coordinates of  \eqref{change:time:scale} first before using averaging method. To to that, we introduce an auxiliary system  
\begin{align}\label{periodic}
	{\frac{d \hat x}{d s}=P(s) \hat x},
\end{align}
and let $\Phi(s,s_0)$ be its state transition matrix. Since $P$ is block-diagonal, it holds that
 $
 	\Phi=\blkdiag(\Phi\PC{1},\dots,\Phi\PC{r}),
$
 where $\Phi\PC{k}$ is the transition matrix of the subsystem in the $k$th cluster ${d \hat x_k}/{d s}=P\PC{k}(s) \hat x_k$. 
 
 Consider the change of coordinates $z(s)=\Phi^{-1}(s,s_0) x(s)$. It follows from the system \eqref{change:time:scale} that
\begin{align}\label{coordinated}
	{\frac{dz}{ds} = \varepsilon (\Phi^{-1} J \Phi z+ \Phi^{-1} f_\inter( \Phi z, y))}.
\end{align}
Since $P(s)$ is $T$-periodic, $\Phi$ and $\Phi^{-1}$ are also $T$-periodic.
Then, we associate \eqref{coordinated} with a partially averaged system
\begin{align}\label{average}
	{\frac{dz}{ds} = \varepsilon (\bar J z+ \Phi^{-1} f_\inter( \Phi z, y))},
\end{align}
where
\begin{align}\label{J-bar}
	{\bar J= \frac{1}{T}\int_{s_0}^{s_0+T} \Phi^{-1} (s,s_0) J \Phi(s,s_0)ds}.
\end{align}

As both $J$ and $\Phi$ are block-diagonal, one can derive that $\bar J$ is also block-diagonal satisfying $\bar J=\blkdiag(\bar J\PC{1},\dots,\bar J\PC{r})$ with 
\begin{align*}
	\bar J \PC{k}= \frac{1}{T}\int_{s_0}^{s_0+T} \Big(\Phi\PC{k}(s,s_0)\Big)^{-1} J\PC{k} \Phi\PC{k}(s,s_0)ds.
\end{align*} 
Recall that $\Phi$ and $\Phi^{-1}$ are periodic, $\|\Phi\|$ and $\|\Phi^{-1}\|$ are both bounded. Then, it can be shown that there exist $\bar \gamma_{k\ell}>0,k,\ell=1,\dots,r$, such that 
\begin{align*}
	\|\Phi^{-1} f\PC{k}_\inter( \Phi z, y)\|\le \sum_{\ell=1}^{r}\bar \gamma_{k\ell} \|z_\ell\|
\end{align*}
for each $k$ (see Lemma~\ref{lemma:bound:pert} in Appendix~\ref{proof:general} for more details). 

\begin{theorem}[{Sufficient condition for vibrational stabilization}]\label{general}
	Assume that $\bar J=\blkdiag(\bar J\PC{1},\dots,\bar J\PC{r})$ in Eq.~\eqref{J-bar} is Hurwitz. Let $\bar X_k$ be the solution to the Lyapunov equation 
	\begin{align}\label{Ly:controlled}
		{(\bar J\PC{k})^\top \bar X_k+\bar X_k^\top \bar J\PC{k}=-I}.
	\end{align}
	Define the matrix $S=[s_{k \ell}]_{r\times r} $ with
	\begin{align}\label{matrix:S}
	s_{k\ell}=\Big\{\begin{matrix*}[l]
		\lambda_{\max}^{-1}(\bar X_k)-\bar\gamma_{k k} , &\text{ if }k=\ell,\\
		s_{k\ell}= -\bar \gamma_{k \ell}, &\text{ if }k \neq\ell.
	\end{matrix*}
	\end{align}
	If $S$ is an $M$-matrix, then  there exists $\varepsilon^*>0$ such that, for any $\varepsilon<\varepsilon^*$:
	
	(i) the equilibrium $x=0$ of the system \eqref{compact_form} is exponentially stable uniformly in $y$; 
	
	(ii) the cluster synchronization manifold $\CM$ of the system \eqref{model:controlled} is exponentially stable. 
\end{theorem}

The proof can be found in Appendix~\ref{proof:general}.
Theorem~\ref{general} provides a guideline to design vibrational control. Any vibrational input that satisfies the following three conditions stabilizes the cluster synchronization manifold $\CM$: 
i) $\bar J$ in \eqref{J-bar} is Hurwitz, 
ii) $S$ defined in \eqref{matrix:S} is an $M$-matrix, and
iii) the frequency of the vibrations is sufficiently high, i.e., $\varepsilon>0$ is sufficiently small.

We shall point out that the condition in Theorem~\ref{general} is still conservative. An example will be presented in Sec.~\ref{numerical} to show  that a vibrational control can effectively stabilize cluster synchronization without satisfying this condition.

\begin{remark}[\textit{Synchronization robustness}]
	For a stable linear system $\dot x=Ax$, previous studies propose that $\lambda_{\max}(X)$, where $X$ is the solution to the Lyapunov equation $A^\top X+X A=-2I$, can measure the robustness of this system \cite{PRV-TM:80,YR:85}. 
	In our case, from \eqref{linearized}, the uncontrolled intra-dynamics around the manifold $\CM$ are described by 
	\begin{align}\label{uncontrolled}
		&\dot x_k=J\PC{k} x_k+f\PC{k}_\inter(x,y),&k=1,\dots,r,
	\end{align}
	where $J\PC{k}$ is stable and $f\PC{k}_\inter(x,y)$ is taken as the perturbation. Here, $x_k=0$ means synchronization of the oscillators in the $k$th cluster. Let $X_k$ be the solution to 
	\begin{align}\label{Ly:uncontrolled}
		\scalemath{0.85}{( J\PC{k})^\top X_k+X_k^\top  J\PC{k}=-I}.
	\end{align}
Likewise, one can use $\lambda^{-1}_{\max}(X_k)$ to measure the robustness of synchronization in the $k$th cluster. If the intra-cluster synchronization is sufficiently robust (i.e., large $\lambda^{-1}_{\max}(X_k)$'s) to dominate the perturbations resulted from inter-cluster connections, the cluster synchronization is stable. A sufficient condition is constructed in \cite[Th. 3.2]{TM-GB-DSB-FP:18}. By contrast, if $\lambda_{\max}(X_k)$'s are not large enough, the cluster synchronization can lose its stability. Yet, the robustness of the intra-cluster synchronization can be reshaped by introducing vibrations to the local network connections. The new robustness is instead measured by $\lambda^{-1}_{\max}( \bar X_k)$, where $\bar X_k$ is the solution to Eq.~\eqref{Ly:controlled}.  \QEDA
\end{remark}
 
The following example illustrates how vibrations can improve synchronization robustness. 

\begin{example}[\textit{Improving robustness by vibrational control}]
We restrict our attention to only one cluster in this example. Specifically, consider that $J\PC{k}$ in \eqref{uncontrolled} is 
	\begin{align*}
		\scalemath{0.85}{J\PC{k}=\begin{bmatrix}
		-1& 4\\0 &-2
	\end{bmatrix}}. 
\end{align*}
Suppose that there is a vibrational control $U\PC{k}(t)$ to the $k$th cluster that results in the state transition matrix 
	\begin{align*}
		{\Phi\PC{k}(t,t_0)=\begin{bmatrix}
				1&0\\
				\cos(t_0)-\cos(t)&1
		\end{bmatrix}}.
	\end{align*} It can be calculated that 
\begin{align*}
	\scalemath{0.9}{\bar J\PC{k} =\frac{1}{2\pi}\int_{0}^{2 \pi} \big(\Phi\PC{k}(t,t_0)\big)^{-1} J\PC{k} \Phi\PC{k}(t,t_0)dt=\begin{bmatrix}
		-1& 4\\-2 &-2
	\end{bmatrix}}.
\end{align*}
Solving the equations \eqref{Ly:controlled} and \eqref{Ly:uncontrolled}, one can obtain $\lambda^{-1}_{\max}( X_k)\approx 0.52$ and $\lambda^{-1}_{\max}(\bar X_k)=2$, which means that the synchronization in this cluster becomes more robust due to the vibrational control. \QEDA
\end{example}

In fact, whether the synchronization robustness can be improved also depends on the intra-cluster network. The following lemma provides a necessary condition. 

\begin{lemma}[\textit{Necessary condition for robustness improvement}]
	Let $\lambda^{-1}_{\max}(\bar X_k)$ and $\lambda^{-1}_{\max}( X_k)$ be the solutions to the equations \eqref{Ly:controlled} and \eqref{Ly:uncontrolled}, respectively. There exists a vibrational control input such that $\lambda^{-1}_{\max}(\bar X_k)> \lambda^{-1}_{\max}(X_k)$ \emph{only if} $\CG_k$ is not a uniformly weighted complete graph\footnote{A uniformly weighted complete graph is a complete graph in which every edge has the same weight.}.
\end{lemma}

 \begin{figure}[t]
	\centering
	\includegraphics[scale=1.3]{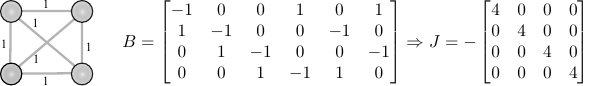}
	\caption{Uncontrollable cluster: a complete and uniformly weighted network. }
	\label{unctr}
\end{figure}

When $\CG_k$ is complete and uniformly weighted, it holds that $J_k=-n_kI_{n_k-1}$, where $n_k$ is the number of nodes (see an example in Fig.~\ref{unctr}). As a consequence, $\lambda^{-1}_{\max}(X_k)= \lambda^{-1}_{\max}(\hat X_k)$ for any vibrational control since
\begin{align*}
	\scalemath{0.9}{\bar J\PC{k} =\int_{0}^{2 \pi} \Phi^{-1}(t) J\PC{k} \Phi(t)=-n_k\int_{0}^{2 \pi} \Phi^{-1}(t) \Phi(t)dt}=J\PC{k}.
\end{align*}

If all the clusters in a network are complete and uniformly weighted, one can see that the synchronization robustness cannot be improved by any vibrational control.

\section{Configuring  Vibrational Control}

\subsection{Configuring lower  triangular transition matrices}

Notice from Theorem~\ref{general} that, when designing a vibrational control, the state transition matrix $\Phi$ of the periodic system \eqref{periodic} plays an important role. However, deriving an explicit expression of $\Phi$ is often not  straightforward. In this subsection, we provide a special approach to configure vibrational control so that the expression of $\Phi$ can be obtained explicitly.

Inspired by \cite{SMM:80}, we consider vibration control inputs $U(t)$ such that $P(t)=\blkdiag (P\PC{1}(t),\dots,P\PC{r}(t))$ in Eq.~\eqref{expres:JandP} satisfies that, for any $k$, $P\PC{k}$ has the following strictly lower triangular form (where $\bar n_k=n_k-1$):
	\begin{align}\label{P(t)}
	&\scalemath{.8}{P\PC{k}(t)=}\nonumber\\
	&\scalemath{0.75}{\begin{bmatrix}
		0 &0 &0 &\cdots &0 &0\\
		u_{21}^{(k)}\sin(t) &0 &0& \cdots &0 &0\\
		u_{31}^{(k)}\sin(t) &u_{32}^{(k)}\sin(t) &0 &\cdots &0 &0\\
		\vdots&\vdots &\vdots&\ddots&\vdots&\vdots\\
		u_{\bar n_k, 1}^{(k)}\sin(t) &u_{\bar n_k,2}^{(k)}\sin(t) &u_{\bar n_k,3}^{(k)}\sin(t) &\cdots &u_{\bar n_k,\bar n_k -1}^{(k)}\sin(t) &0			
	\end{bmatrix}}.
\end{align}

\begin{lemma}[\textit{Lower triangular transition matrix}]
	Assume that the vibrational control is such that $P(t)$  in the system~\eqref{expres:JandP}  satisfies \eqref{P(t)}. Then, the state transition matrix $\Phi(t,t_0)=\blkdiag(\Phi_1,\dots,\Phi_r)$ satisfies
	\begin{align*}
		\scalemath{0.9}{\Phi_k=\begin{bmatrix}
				1 &0 &0 &\cdots &0 &0\\
				\Phi_{21}^{(k)} &1 &0& \cdots &0 &0\\
				\Phi_{31}^{(k)} &\Phi_{32}^{(k)} &1 &\cdots &0 &0\\
				\vdots&\vdots &\vdots&\ddots&\vdots&\vdots\\
				\Phi_{\bar n_k,1}^{(k)} &\Phi_{\bar n_k,2}^{(k)} &\Phi_{\bar n_k,3}^{(k)} &\cdots &\Phi_{\bar n_k,\bar n_k-1}^{(k)} &1			
		\end{bmatrix}},
	\end{align*}
where $\scalemath{1}{\Phi_{ij}^{(k)}(t)=-u_{ij}^{(k)}(\cos (t)-\cos(t_0))}$.
\end{lemma}

The next question is whether there exist vibrational control inputs such that $P(t)$ has the form described by \eqref{P(t)}.
\begin{lemma}[\textit{Existence of vibrational control}]\label{lemma:existence}
	For each intra-cluster network $\CG_k,k=1,\dots,r$, there always exist vibrational control inputs such that $P\PC{k}(t)$ is strictly lower triangular and satisfies $P\PC{k}(t)\neq 0$. 
\end{lemma}

The proof can be found in Appendix~\ref{proof:existence}.
From this lemma, one can always carefully configure the vibrational control so that $P(t)$ is strictly lower triangular. Subsequently, the explicit form of the transition matrix $\Phi$ can be derived, making the design of vibrations more analytically tractable. 

It is worth noting that there certainly exist other configurations of vibration control different from the one we present in this subsection, which also allow for easy derivation of explicit transition matrices.

\subsection{A numerical example}\label{numerical}
In this subsection, we use an example to show how the obtained results can be used to design vibrational control.

We consider the network depicted in Fig.~\ref{ctr}~(a), where the oscillators are partitioned into $3$ clusters.  
The matrix $J=\blkdiag(J\PC{1},\dots,J\PC{3})$ as in the system \eqref{linearized} satisfies
\begin{align*}
	{J\PC{1}=\begin{bmatrix}
		-0.06&0.06\\
		0&-0.18
	\end{bmatrix},J\PC{2}=-3I_3,\text{ and }J\PC{3}=-6I_6.}
\end{align*}
It can be seen in Fig.~\ref{ctr}~(b) that the cluster synchronization manifold $\CM$ in this network is unstable. 

 \begin{figure}[t]
	\centering
	\includegraphics[scale=1.3]{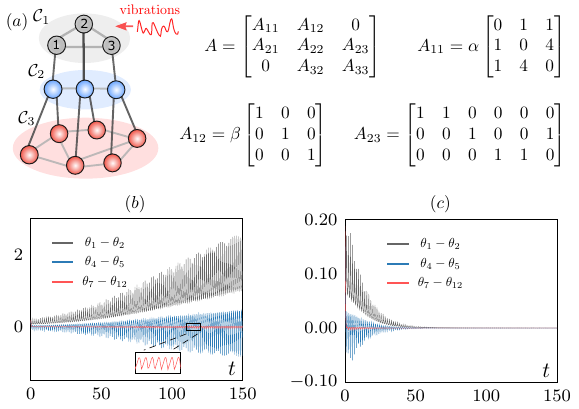}
	\caption{Vibrational stabilization. (a) The network structure and some connection weights, where $\alpha=0.03$ and $\beta=3$. All other unspecified weights are $1$. (b) Phase differences without control, indicating that the cluster synchronization is unstable. (c) Phase differences under vibrational control to the cluster $\CC_1$, showing that the cluster synchronization has been stabilized by just local vibrations ($\varepsilon=0.02$). The natural frequencies in $\CC_1$, $\CC_2$, and $\CC_3$  are $\omega_1=1$, $\omega_2=10$, and $\omega_3=6$, respectively.}
	\label{ctr}
\end{figure}

To stabilize the cluster synchronization manifold, we only introduce vibrations to the cluster $\CC_1$. Specifically, we let
\begin{align*}
	\begin{matrix*}[l]
		&u_{12}(t)=\scalemath{0.8}{\frac{u}{2\varepsilon}} \sin(\EPt),&u_{13}(t)=-\scalemath{0.8}{\frac{u}{2\varepsilon}} \sin(\EPt), \\
		&u_{32}(t)= -\scalemath{0.8}{\frac{u}{2\varepsilon}} \sin(\EPt), &u_{31}(t)=\scalemath{0.8}{\frac{u}{2\varepsilon}} \sin(\EPt),
	\end{matrix*}
\end{align*}
and $u_{ij}=0$ for any other $i$ and $j$.  This configuration of vibrational control ensures that the matrix $P\PC{1}(s)$ in Eq.~\eqref{expres:JandP} is strictly lower triangular, i.e.,
\begin{align*}
	\scalemath{.85}{P\PC{1}(s)=\begin{bmatrix}
			0&0\\
			-u\sin(s)&0
	\end{bmatrix}}.
\end{align*}
Then, the transition matrix of the periodic system \eqref{periodic} can be calculated as
\begin{align*}
\scalemath{0.85}{\Phi\PC{1}(s,s_0)=\begin{bmatrix}
		1&0\\
		u(\cos(s_0)-\cos(t))&1
\end{bmatrix}}.
\end{align*}
Subsequently, the matrix $\bar J\PC{1}$ in Eq.~\eqref{J-bar} is
\begin{align*}
	\scalemath{.85}{\bar J\PC{1}= \frac{1}{2\pi}\int_{0}^{2\pi} \big(\Phi\PC{1}(s,s_0)\big)^{-1} J\PC{1} \Phi\PC{1}(s,s_0)ds=\begin{bmatrix}
		-0.06&0.06\\
		\frac{-0.06u^2}{8}&-0.18
	\end{bmatrix}}.
\end{align*}

Let $u=2\sqrt{2}$. Solving the equations \eqref{Ly:controlled} and \eqref{Ly:uncontrolled}, respectively, one obtains $\lambda^{-1}_{\max}( X_1)\approx 0.109$ and $\lambda^{-1}_{\max}(\bar X_1)\approx 0.133$. This indicates that the synchronization robustness in $\CC_1$ is improved by the vibrational control. From Fig.~\ref{ctr}~(c), it can be observed that the local vibrations introduced to $\CC_1$, with frequency $50$ rad/s (i.e., $\varepsilon=0.02$), successfully stabilize the cluster synchronization in the entire network. 

It is worth mentioning that the condition in Theorem~\ref{general} is not satisfied. Yet and  remarkably, the slight improvement on synchronization robustness still stabilizes the cluster synchronization. This suggests that vibrational control can practically work in a much broader region than the theoretical one we identified in Theorem~\ref{general}.


\section{Conclusions}
In this work, we hypothesize that the vibrational control is a mechanism of DBS, a widely-used neurosurgical technique to treat some brain disorders, based on the observation that they both employ high-frequency dithers and regulate synchronization patterns. 
We study how the vibrational control can stabilize cluster synchronization in Kuramoto-oscillator networks. Some sufficient conditions on the vibrations introduced to the network connections are obtained. We also provide an analytically  tractable approach to design vibrational inputs. Finally, we demonstrate our theoretical findings by a numerical example. We believe that our work provides a new angle to understand DBS and can inform the design of DBS algorithms.

\appendix

\subsection{Derivaton of Eq.~\eqref{compact_form}} \label{sec:compact_form}

To derive Eq.~\eqref{compact_form}, we first present an instrumental lemma.
\begin{lemma}\label{incidence:transfer}
	For the incidence matrices $B$ and $\hat B$, there exists $R=\SmallMatrix{
		R_1&0\\
		R_2&R_3}$ such that $B^\top =R\hat B^\top$,  where 
	\begin{align*}
		\scalemath{0.85}{R_1= \blkdiag\big((B^{(1)}_{\intra})^\top ((\hat B^{(1)}_{\intra})^\top)^\dagger,\dots,(B^{(r)}_{\intra})^\top ((\hat B^{(r)}_{\intra})^\top)^\dagger\big),}
	\end{align*}
	and $\SmallMatrixNBla{R_2=B_\inter^\top (\hat B_\intra ^\top P_\inter)^\dagger}$,
	$\SmallMatrixNBla{R_3= B_\inter^\top (\hat B_\inter P_\intra)^\dagger}$
with $\SmallMatrixNBla{P_\intra=I_n-\hat B_\intra \hat B_\intra^\dagger}$ and $\SmallMatrixNBla{P_\inter=I_n-\hat B_\inter \hat B_\inter^\dagger}$. \QEDA
\end{lemma}
\begin{proof}
	The expressions of $R_2$ and $R_3$ follow from \cite{KR-IH:2021}. Recall that ${B_\intra=\blkdiag( B_{\intra}^{(1)},\dots, B_{\intra}^{(r))}}$and ${\hat B_\intra=\blkdiag(\hat B^{(1)}_{\intra},\dots,\hat B^{(r)}_{\intra})}$. Since ${(B^{(k)}_{\intra})^\top }$ has full row rank, we have 
	\begin{align*}
		(B^{(k)}_{\intra})^\top = (B^{(k)}_{\intra})^\top ((\hat B^{(k)}_{\intra})^\top)^\dagger(\hat B^{(k)}_{\intra})^\top,
	\end{align*}
	which completes the proof.
\end{proof}

One can rewrite Eq.~\eqref{model:controlled} into a compact form 
\begin{align*}
	\dot \theta = \omega-BW \sin(B^\top \theta)-B_\dir^+ U \sin(B_\dir^\top \theta).
\end{align*}
From Lemma~\ref{incidence:transfer}, it holds that $ B^\top \theta=R \hat B^\top \theta$ and $B_\dir^\top \theta=R[ \hat B_\intra,-\hat B_\intra]^\top \theta$. Subsequently, one can derive
\begin{align*}
	\hat B^\top \dot \theta=& \hat B^\top \omega-\hat B^\top BW \sin(R \hat B^\top \theta)\\
	&-B_\dir^+ U \sin(R[ \hat B_\intra,-\hat B_\intra]^\top \theta).
\end{align*}

Since $x=\hat B_\intra^\top \theta$ and $y = \hat B_\inter^\top \theta$, one can obtain the compact dynamics \eqref{compact_form} using the fact that $[B_{\dir1}]^+=[B_\intra]^+$ and  $[B_{\dir1}]^+=[-B_\intra]^+$ and the relations between matrices in  Eq.~\eqref{matrices} .

\subsection{Proof of Lemma~\ref{lemma:linearized}}\label{proof:linearized}
Since $x=0$ is exponentially stable uniformly in $y$ for the system \eqref{linearized}, according to the converse Lyapunov theorem  (see \cite[Th. 4.4]{haddad2011nonlinear} and \cite{QY-KY-BDOA-CM:2021}) there exists $\CD=\{x\in\R^{n-r}:\|x\|\le \rho_1\}$ and a  continuously
differentiable function function $V:[0,\infty]\times \CD\times \R^r\to \R $ such that 
\begin{align*}
	\frac{\partial V}{\partial t}+\frac{\partial V}{\partial x}\big( (J + P(t))x+f_\inter(x,y)\big)\le -c_1\|x\|^2
\end{align*}
and $\|\frac{\partial V}{\partial x}\|\le c_2\|x\|$ for some constants $c_1,c_2> 0$.  
Let 
$
	h(t,x)=f_\intra(x)+f_\ctr(U(t),x)
$
and 
$
	\Delta(t,x)=h(t,x)-(J + P(t))x.
$ 
It can be checked that $\partial h/\partial x$ is bounded and Lipschitz on $\CD$. Then, similar to the proof of \cite[Th. 4.13]{HKK:02-bis}, one can show that $\|\Delta(t,x)\|\le c_3 \|x\|^2$ for some $c_3>0$. The time derivative along the system  \eqref{compact_form} satisfies
\begin{align*}
	&{\frac{\partial V}{\partial t}+\frac{\partial V}{\partial x}\big( (J + P(t))x+\Delta(t,x)+f_\inter(x,y)\big)}\\
	&\le -c_1\|x\|^2+c_2c_3\|x\|^3\\
	&\le -(c_1-c_2c_3\rho)\|x\|^2, \forall \|x\|< \rho.
\end{align*}
Choosing $\rho=\min\{\rho_1,c_1/c_2c_3\}$ completes the proof.

\subsection{Proof of Theorem~\ref{general}} \label{proof:general}
 One can observe that (i) implies (ii). Then, it suffices to prove the exponential stability of $x=0$ for \eqref{change:time:scale}. To do that, we first present the following lemma, whose proof follows similar lines as Lemma 3.1 of \cite{TM-GB-DSB-FP:18}.
\begin{lemma}[Growth bound of perturbations]\label{lemma:bound:pert}
	There exist some constants $\bar \gamma_{k\ell}>0$, $k,\ell=1,\dots,r$, such that, for any $k$, it holds that 
	\begin{align*}
	\|\Phi^{-1} f\PC{k}_\inter( \Phi z, y)\|\le \sum\nolimits_{\ell=1}^{r}\bar \gamma_{k\ell} \|z_\ell\|.\hspace{20pt}\QEDA
	\end{align*} 
\end{lemma}

Let $V_k=z_k^\top X_kz_k$, and it holds that $\partial V_k/\partial z_k\le \lambda_{\max}(X_k)$. Choose $V(z)= \sum_{k=1}^{r}d_k V_k$ as a Lyapunov candidate. The time derivative of $V(z)$ satisfies
\begin{align*}
	\scalemath{0.85}{\dot V(z)}&=\scalemath{0.85}{\sum_{k=1}^{r} d_k[z_k^\top((\bar J\PC{k})^\top X_k+X_k \bar J\PC{k})z_k+\frac{\partial V}{\partial z_k}\Phi^{-1} f_\inter( \Phi z, y)]}\\
	&\scalemath{0.85}{\le \sum_{k=1}^{r} d_k[-\|z_k\|^2 + \lambda_{\max}(X_k)\sum_{k=1}^{r}\bar \gamma_{k \ell} \|z_k\|\|z_\ell\| ]},
\end{align*}
where the second inequality has used Lemma~\ref{lemma:bound:pert}.

Let ${D:=\diag(d_1,\dots,d_r)}$ and ${\hat S=[\hat s_{ij}]_{r\times r}}$ where 
\begin{align*}
	\hat s_{k\ell}=\Big\{\begin{matrix*}[l]
		1- \lambda_{\max}(X_k) \bar \gamma_{k k} , &\text{ if }k=\ell,\\
	\hat s_{k\ell}= -\lambda_{\max}(X_k)\bar \gamma_{k \ell}, &\text{ if }k \neq\ell.
	\end{matrix*}
\end{align*}
Then, one can rewrite ${\dot V(z)\le -\frac{1}{2}z^\top (DS+S^\top D)z}$. By assumption, $S$ is an $M$-matrix, and so is $\hat S$ since $\hat S=-\lambda_{\max}(X_k)S$. It follows from \cite[Th. 9.2]{HKK:02-bis} that the system \eqref{average} is exponentially stable.  

Following similar steps as in \cite[Th. 10.4]{HKK:02-bis}, one can prove that there exists $\varepsilon^*>0$ such that for any $\varepsilon<\varepsilon^*$, $z=0$ is exponentially stable uniformly in $y$ for the system \eqref{coordinated}. Since $x(s)=\Phi(s,s_0)z(s)$ and $\|\Phi\|$ is bounded, then $x=0$ is also exponentially stable uniformly in $y$ for \eqref{change:time:scale}, which completes the proof. 

\subsection{Proof of Lemma~\ref{lemma:existence}}\label{proof:existence}

\begin{proof} We construct the proof by considering a specific configuration of vibrational control. 
	
	Without loss of generality, we assume that each $B_{\intra}^{(k)}$ and $\hat B_{\intra}^{(k)}$ are ordered in a way such that their first $\bar n_k=n_k-1$ columns are the same, and the first two column are $e_2-e_1$ and $e_3-e_2$, respectively, where $e_i$ is the $i$th column of $I_{ n_k}$. Then, there exists a matrix $\tilde{R}^{(k)}_1$ such that 
	$
		R_1^{(k)}={\begin{bmatrix}
				I_{\bar n_k}\\ \tilde{R}^{(k)}_1
		\end{bmatrix}}.
	$

Let $U^{(k)}_{1}(t)=u e_1e_1^\top \sin(t)$ and $U^{(k)}_{2}(t)=-u e_1e_1^\top\sin(t)$.  It can be calculated that 
\begin{align*}
	&P\PC{k}(t)\\
	&= -(\hat B^{(k)}_{\intra})^\top [B^{(k)}_{\intra}]^+  U^{(k)}_{1}(t)-[-B^{(k)}_{\intra}]^{+} U^{(k)}_{2}(t) R_1^{(k)} \\
	&\scalemath{1}{=-(\hat B^{(k)}_{\intra})^\top u(e_1+e_2)e_1^\top \begin{bmatrix}
			I_{\bar n_k}\\ \tilde{R}^{(k)}_1\end{bmatrix}\sin(t)}\\
	&\scalemath{1}{=-u(\hat B^{(k)}_{\intra})^\top (e_1+e_2) (e_1')^\top\sin(t)   },	
\end{align*}
where the second equality follows from $e_1^\top \scalemath{0.7}{\begin{bmatrix}
		I_{\bar n_k}\\ \tilde{R}^{(k)}_1
\end{bmatrix}}=(e_1')^\top $ with $e_1'$ being the first column of $I_{\bar n_k}$. Then, it can be seen that the first two rows of $P\PC{k}(t)$ are 
$
	-u(e_2-e_1)^\top (e_1+e_2) (e_1')^\top \sin(t)=0,
$
and 
$
	-u(e_3-e_2)^\top (e_1+e_2) (e_1')^\top=(e'_1)^\top \sin(t),
$
respectively.  The rest rows of $P\PC{k}(t)$ can only have non-zero values at the first column. Therefore, $P\PC{k}(t)$ have the strictly  lower triangular form and satisfies $P\PC{k}(t)\neq 0$, which completes the proof. 
\end{proof}

\bibliographystyle{IEEEtran}
\bibliography{IEEEabrv,\alias,\FP,\Main,\New}



\end{document}

%% file: notation.tex
\usepackage{mathtools}
\usepackage{bm}
\usepackage{bbm}

\usepackage{graphicx}
\usepackage{amssymb}
\usepackage{color}

\newcommand{\alias}{C:/Users/yuzhe/Dropbox/bib/alias}
\newcommand{\New}{C:/Users/yuzhe/Dropbox/bib/New}
\newcommand{\Main}{C:/Users/yuzhe/Dropbox/bib/Main}
\newcommand{\FP}{C:/Users/yuzhe/Dropbox/bib/FP}

\newcommand{\BS}{\mathbb{S}} 
\newcommand{\R}{\mathbb{R}} 
\newcommand{\BT}{\mathbb{T}} 

\newcommand{\CS}{\mathcal{S}} 
\newcommand{\CC}{\mathcal{C}} %
\newcommand{\CG}{\mathcal{G}} %
\newcommand{\CE}{\mathcal{E}} %
\newcommand{\CM}{\mathcal{M}} %
\newcommand{\CV}{\mathcal{V}} %
\newcommand{\CD}{\mathcal{D}} %

\newcommand{\diag}{{\rm diag}} 

\usepackage{amsthm}

\newtheorem{theorem}{Theorem}

\newtheorem{lemma}{Lemma}

\theoremstyle{remark}
\newtheorem{remark}{Remark}
\newtheorem{example}{Example}
\newtheorem{assumption}{Assumption}
\newtheorem{definition}{Definition}